\documentclass[letterpaper,conference]{IEEEtranBB}   % conference, peerreviewca, twocolumn, 11pt
\usepackage[T1]{fontenc}
\usepackage{graphicx}
	\graphicspath{{figs/}{matlab/}}   % still needed, even with Inkscape-Pdf-Latex
\usepackage[cmex10]{amsmath}
\usepackage{amssymb, amsthm, gensymb}  % gensymb for \degree
\usepackage{booktabs}  % DECISION: don't use colored rows
	\renewcommand{\tabcolsep}{4mm}  % twice this is the tabular column separation (default 6pt)
\usepackage{cellspace}
\usepackage{float}
\usepackage{flafter}  % This to put all floats strictly AFTER they are first referenced
\usepackage{accents}
\usepackage{cases}  % for \begin{cases} with individual eqn numbers: use \begin{numcases}
\usepackage[usenames,dvipsnames]{color}

\usepackage{calc}

\newcommand{\scrs}{\scriptsize}

\newcommand{\executeiffilenewer}[3]{%
\ifnum\pdfstrcmp{\pdffilemoddate{#1}}%
{\pdffilemoddate{#2}}>0%
{\immediate\write18{#3}}\fi%
}
\newcommand{\includesvg}[1]{%
\executeiffilenewer{#1.svg}{#1.pdf}%
{inkscape -z -D --file=#1.svg %
--export-pdf=#1.pdf --export-latex}%
\ifx\svgscale\undefined%
	\newcommand{\svgscale}{\defaultsvgscale}
\fi%
\input{#1.pdf_tex}%
}

	\newcommand{\defaultsvgscale}{1.0}  

\usepackage[tight,small,bf]{subfigure}   % small vs footnotesize
\usepackage{microtype}

\usepackage{xspace}
\usepackage{bbm}
\usepackage{mathrsfs}
\catcode`@=11
\chardef\mathlig@atcode\count255

\def\actively#1#2{\begingroup\uccode`\~=`#2\relax\uppercase{\endgroup#1~}}
\def\mathlig@gobble{\afterassignment\mathlig@next@cmd\let\mathlig@next= }
\def\mathlig@delim{\mathlig@delim}
\def\mathlig@defcs#1{\expandafter\def\csname#1\endcsname}
\def\mathlig@let@cs#1#2{\expandafter\let\expandafter#1\csname#2\endcsname}
\def\mathlig@appendcs#1#2{\expandafter\edef\csname#1\endcsname{\csname#1\endcsname#2}}

\def\mathlig#1#2{\mathlig@checklig#1\mathlig@end\mathlig@defcs{mathlig@back@#1}{#2}\ignorespaces}

\def\mathlig@checklig#1#2\mathlig@end{%
 \expandafter\ifx\csname mathlig@forw@#1\endcsname\relax
 \expandafter\mathchardef\csname mathlig@back@#1\endcsname=\mathcode`#1%
 \mathcode`#1"8000\actively\def#1{\csname mathlig@look@#1\endcsname}%
 \mathlig@dolig#1\mathlig@delim
\fi
\mathlig@checksuffix#1#2\mathlig@end
}

\def\mathlig@checksuffix#1#2\mathlig@end{%
\ifx\mathlig@delim#2\mathlig@delim\relax\else\mathlig@checksuffix@{#1}#2\mathlig@end\fi
}
\def\mathlig@checksuffix@#1#2#3\mathlig@end{%
\expandafter\ifx\csname mathlig@forw@#1#2\endcsname\relax\mathlig@dosuffix{#1}{#2}\fi
\mathlig@checksuffix{#1#2}#3\mathlig@end
}

\def\mathlig@dosuffix#1#2{%
\mathlig@appendcs{mathlig@toks@#1}{#2}%
\mathlig@dolig{#1}{#2}\mathlig@delim
}

\def\mathlig@dolig#1#2\mathlig@delim{%
 \mathlig@defcs{mathlig@look@#1#2}{%
 \mathlig@let@cs\mathlig@next{mathlig@forw@#1#2}\futurelet\mathlig@next@tok\mathlig@next}%
 \mathlig@defcs{mathlig@forw@#1#2}{%
  \mathlig@let@cs\mathlig@next{mathlig@back@#1#2}%
  \mathlig@let@cs\checker{mathlig@chck@#1#2}%
  \mathlig@let@cs\mathligtoks{mathlig@toks@#1#2}%
  \expandafter\ifx\expandafter\mathlig@delim\mathligtoks\mathlig@delim\relax\else
  \expandafter\checker\mathligtoks\mathlig@delim\fi
  \mathlig@next
 }%
 \mathlig@defcs{mathlig@toks@#1#2}{}%
 \mathlig@defcs{mathlig@chck@#1#2}##1##2\mathlig@delim{%
  \ifx\mathlig@next@tok##1%
   \mathlig@let@cs\mathlig@next@cmd{mathlig@look@#1#2##1}\let\mathlig@next\mathlig@gobble
  \fi 
  \ifx\mathlig@delim##2\mathlig@delim\relax\else
   \csname mathlig@chck@#1#2\endcsname##2\mathlig@delim
  \fi
 }%
 \ifx\mathlig@delim#2\mathlig@delim\else
  \mathlig@defcs{mathlig@back@#1#2}{\csname mathlig@back@#1\endcsname #2}%
 \fi
}%

\catcode`@\mathlig@atcode

\newcommand{\muspace}{\mspace{1mu}}

\DeclareRobustCommand{\scond}{\mathchoice{\muspace\vert\muspace}{\vert}{\vert}{\vert}}
\mathlig{|}{\scond}
\newcommand{\cond}{\mathchoice{\,\vert\,}{\mspace{2mu}\vert\mspace{2mu}}{\vert}{\vert}}

\DeclareRobustCommand{\discint}{\mathchoice{\mspace{-1.5mu}:\mspace{-1.5mu}}{\mspace{-1.5mu}:\mspace{-1.5mu}}{:}{:}}
\mathlig{::}{\discint}

\def\Var{\mathop{\rm Var}\nolimits}%
\newcommand{\Cc}{\mathcal{C}}

\def\eps{\epsilon}

\DeclareMathOperator\E{\textsf{E}}
\let\P\relax
\DeclareMathOperator\P{\textsf{P}}

\newcommand{\N}{\mathrm{N}}
\newcommand{\U}{\mathrm{Unif}}

\def\textiid{i.i.d.\@\xspace}
\newcommand\iid{\ifmmode\text{ i.i.d. } \else \textiid \fi}

\newcommand{\Real}{\mathbb{R}}

\def\mathllap{\mathpalette\mathllapinternal}
\def\mathllapinternal#1#2{%
  \llap{$\mathsurround=0pt#1{#2}$}}

\def\clap#1{\hbox to 0pt{\hss#1\hss}}
\def\mathclap{\mathpalette\mathclapinternal}
\def\mathclapinternal#1#2{%
  \clap{$\mathsurround=0pt#1{#2}$}}

\let\oldstackrel\stackrel
\renewcommand{\stackrel}[2]{\oldstackrel{\mathclap{#1}}{#2}}

\renewcommand{\hbar}{h\mathllap{\overline{\vphantom{h}\hphantom{\rule{4.6pt}{0pt}}}\mspace{0.77mu}}}

\catcode`~=11 % make LaTeX treat tilde (~) like a normal character
\newcommand{\urltilde}{\kern -.06em\lower -.06em\hbox{~}\kern .02em}
\catcode`~=13 % revert back to treating tilde (~) as an active character

\usepackage{caption}
\DeclareCaptionLabelSeparator{periodquad}{. \quad}
\ifCLASSOPTIONonecolumn
	\ifCLASSOPTIONtrimmargin
		\usepackage{anysize}
		\papersize{9.7in}{5.9in}  % hw
		\marginsize{0.2in}{0.2in}{0.1in}{0.1in} % lrtb
	\else
	\fi
\fi

\usepackage[colorlinks=true,linkcolor=black,anchorcolor=black,citecolor=black,filecolor=black,menucolor=black,urlcolor=black]{hyperref}  % for PDF meta information, see below
	
\usepackage{cite}  % has to be after hyperref for pdf-crossreferences to work with citations
\newcommand{\T} {^\mathrm{T}}

\renewcommand{\epsilon}{\varepsilon} %BB
\renewcommand{\eps}{\varepsilon}
\newcommand{\abs}[1]{\lvert #1 \rvert}

\renewcommand{\det}[1]{\lvert #1 \rvert}

\providecommand{\abs}[1]{\lvert#1\rvert}

\definecolor{unemphColor}{gray}{0.5}  %FIXME: do we want unemph here or not?

\newcommand{\annleq}[1]{\overset{\text{(#1)}}{\leq}}  % annotated <=
\newcommand{\anngeq}[1]{\overset{\text{(#1)}}{\geq}}  % annotated >=
\newcommand{\natSet}[1]{\{1::#1\}}  % {1:n}
\renewcommand{\U}{\mathrm{Unif}}
\newcommand{\Typ}{\mathcal T_\varepsilon^{(n)}}
\newcommand{\Typprime}{\mathcal T_{\varepsilon'}^{(n)}}

\theoremstyle{definition}  % avoid italic-set theorems
\newtheorem{thm}{Theorem}

\theoremstyle{remark}
\newtheorem{remark}{Remark}

\newcommand{\Es}{E_S}
\newcommand{\Nb}{\bar N}

\newcommand{\Zb}{\bar Z}
\newcommand{\Zbb}{\bar{ \bar Z}}

\title{Gaussian State Amplification \\with Noisy State Observations}

\IEEEoverridecommandlockouts
\author{
\IEEEauthorblockN{Bernd Bandemer} 
\IEEEauthorblockA{
ITA Center, UC San Diego \\
La Jolla, CA 92093, USA \\
Email: bandemer@ucsd.edu }
\and
\IEEEauthorblockN{Chao Tian} 
\IEEEauthorblockA{
AT\&T Labs-Research \\
Florham Park, NJ 07932, USA \\
Email: tian@research.att.com }
\and
\IEEEauthorblockN{Shlomo Shamai (Shitz)} 
\IEEEauthorblockA{
Dept.~of EE, Technion--IIT \\
Technion City, Haifa 32000, Israel \\
Email: sshlomo@ee.technion.ac.il }

\thanks{\hrule \vspace{2mm} \noindent 
The work of S.~Shamai (Shitz) was supported by the Israel Science Foundation (ISF), and the European Commission in the framework of the Network of Excellence in Wireless COMmunications NEWCOM\#.
}
}

\hypersetup{
	pdfauthor = {Bernd Bandemer, Chao Tian, and Shlomo Shamai (Shitz)},
	pdftitle = {},
	pdfsubject = {Information theory},
	pdfkeywords = {},
	pdfcreator = {LaTeX}%,
}

\begin{document}
\maketitle

%%%%%%%%%%%%%%%%%%%%%%%%%%%%%%%%%%%%%%%%%%%%%%%%%%%%%%%%%%%%%%%%%%%%%%%%%%%%%%%
%%%%%%%%%%%%%%%%%%%%%%%%%%%%%%%%%%%%%%%%%%%%%%%%%%%%%%%%%%%%%%%%%%%%%%%%%%%%%%%
%%%%%%%%%%%%%%%%%%%%%%%%%%%%%%%%%%%%%%%%%%%%%%%%%%%%%%%%%%%%%%%%%%%%%%%%%%%%%%%
% TODO:
% -- In the inner bound, when giving the coding scheme, there should be a comment that this is not plain DPC, because the coefficient is chosen differently.
%%%%%%%%%%%%%%%%%%%%%%%%%%%%%%%%%%%%%%%%%%%%%%%%%%%%%%%%%%%%%%%%%%%%%%%%%%%%%%%
%%%%%%%%%%%%%%%%%%%%%%%%%%%%%%%%%%%%%%%%%%%%%%%%%%%%%%%%%%%%%%%%%%%%%%%%%%%%%%%

\begin{abstract} 
The problem of simultaneous message transmission and state amplification in a Gaussian channel with additive Gaussian state is studied when the sender has imperfect noncausal knowledge of the state sequence. Inner and outer bounds to the rate--state-distortion region are provided. The coding scheme underlying the inner bound combines analog signaling and Gelfand--Pinsker coding, where the latter deviates from the operating point of Costa's dirty paper coding.
\end{abstract}

%\begin{figure}
%	\centering
%	\includesvg{figs/f_3dic}
%	\caption{Three-pair interference channel under consideration.}
%	\label{fig:f_3dic}
%\end{figure}

\vspace{3mm}
\section{Introduction}

Consider a Gaussian channel with additive Gaussian state in which the receiver simultaneously recovers a message communicated by the sender and estimates the state sequence. The sender facilitates this process by utilizing its own (possibly imperfect) knowledge of the state.

This problem was first investigated by Sutivong {\em et al.}~\cite{Sutivong2005} for the case when the sender has perfect knowledge of the state sequence before transmission begins. The authors show that the optimal rate--state-distortion tradeoff is achieved by dividing the available transmit power between analog state transmission and message transmission via dirty paper coding~\cite{Costa1983}. 
Subsequently, Kim {\em et al.} \cite{Kim2008} considered the discrete memoryless version of this problem and characterized the tradeoff when the state reconstruction accuracy is measured by blockwise mutual information instead of quadratic distortion. 

In this work, we are interested in the case in which the state observation at the sender is not perfect, but is encumbered by additive Gaussian state observation noise. 
%In this work, we consider a generalization of the case  in \cite{Sutivong2005}, where only a noisy version of the Gaussian channel states are observed at the encoder, which wishes to simultaneously transmit a message and amplify the real channel states hidden under the observation noise. 
%
A generalization of the problem in \cite{Sutivong2005}, this setting can be understood as modeling the original state amplification system with imperfect processing components. It also applies to the relay channel (see, for example,~\cite{ElGamalKim}), where the relay node attempts to amplify the primary transmitter's signal of which it has obtained a lossy description through another route.
% setting when the transmitter is the relay node, and the channel states are the primary  transmitter's signal, which is  given to the helper transmitter through another route possibly after some lossy compression. 

Our setting contains several interesting extreme cases. 
When the state observation noise becomes negligible, our setting reverts to that of~\cite{Sutivong2005}. Conversely, the case when the state observation noise grows infinite, and thus the transmitter has no state knowledge, was studied in~\cite{Zhang2011}.
Pure state amplification without message transmission was considered in a previous work by the second author~\cite{Tian2012}, where it was shown that an analog scheme with power control is optimal. By contrast, in the case of pure message transmission without state amplification, the optimal rate is achieved by dirty paper coding~\cite{Costa1983} with respect to the observable part of the state.

In this work, we consider the general tradeoff between message transmission rate and state amplification accuracy. We propose an inner bound and two outer bounds to the rate--state-distortion region. 
The inner bound is obtained by a hybrid scheme of analog state signaling and Gelfand--Pinsker encoding, where the Gelfand--Pinsker code in general does not coincide with dirty paper coding~\cite{Costa1983}, but instead requires an optimized coefficient choice.
The first outer bound is derived by generalizing the noise-partition approach in~\cite{Tian2012}, while the second outer bound follows from careful analysis of the correlation structure in the problem. 

In the following, we first provide a precise problem definition, before discussing the inner and outer bounds in Sections~\ref{sec:ib} and~\ref{sec:ob}. We omit some details of the proofs for brevity's sake. 
Numerical examples and concluding remarks are given in Section~\ref{sec:numExConclusion}. Our mathematical notation follows~\cite{ElGamalKim}.

\vspace{3mm}
\section{Problem definition}
Consider the state-dependent memoryless channel with input $X_i \in \Real$ for $i=1,2,\dots$ and output
\begin{align*}
	Y_i &= X_i + S_i + Z_i,
\end{align*}
where $\{S_i\}$ and $\{Z_i\}$ are additive i.i.d.~state and noise sequences, respectively, distributed according to $S_i \sim \N(0,Q)$ and $Z_i\sim \N(0,N)$. The channel input is subject to an average power constraint $P$. The sender has non-causal access to the noisy state observation
\begin{align*}
	V_i &= S_i +U_i,
\end{align*}
where $\{U_i\}$ is an i.i.d.~state observation noise sequence distributed according to $U_i \sim \N(0,\sigma_u^2)$. We assume that $\{S_i\}$, $\{Z_i\}$, and $\{U_i\}$ are independent.

\begin{figure}%[h]  
	\centering
	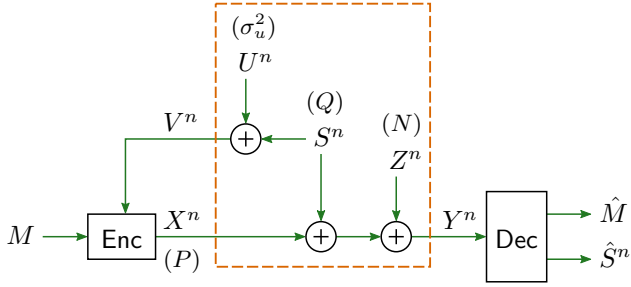
	\caption{State amplification with noisy state observations.}
	\label{fig:stateAmpWithNoisyObs}
\end{figure}

The sender aims to communicate a message $M$ at rate $R$ reliably to the receiver. In addition to recovering the message, the receiver is also interested in obtaining an estimate $\hat S^n$ of the state sequence $S^n$, with state distortion measured by the average squared error and upper-bounded by $D$. The setup is depicted in Figure~\ref{fig:stateAmpWithNoisyObs}.

Formally, a $(2^{nR},n)$ code for the Gaussian state amplification channel with noisy state observations consists of an encoder that maps a message $m \in \natSet{2^{nR}}$ and a state observation $v^n \in \Real^n$ to a transmit sequence $x^n \in \Real^n$ and a decoder that maps the received sequence $y^n \in \Real^n$ to a message estimate $\hat m \in \natSet{2^{nR}}$ and a state estimate $\hat s^n \in \Real^n$. We assume the message is random according to  $M \sim \U \natSet{2^{nR}}$ and restrict our attention to codes that satisfy $(1/n) \sum_{i=1}^n \E[X_i^2]\leq P$.
A rate--state-distortion pair $(R,D)$ is said to be achievable if there exists a sequence of $(2^{nR},n)$ codes such that $\lim_{n\to\infty} \P\{\hat M \neq M\} = 0$ and $\limsup_{n\to\infty} \ (1/n) \sum_{i=1}^n \E(S_i-\hat S_i)^2 \leq D$.
We are interested in characterizing the rate--distortion region, i.e., the closure of the set of achievable $(R,D)$ pairs.

Before we begin our discussion, it is helpful to reformulate the problem as follows. Since the triple $(S_i, U_i, V_i)$ is jointly Gaussian, we can equivalently write
\begin{align*}
	S_i &= \tilde V_i + W_i, 
\end{align*}
with $\tilde V_i = Q/(Q+\sigma_u^2) V_i \sim \N(0,Q')$ and $W_i \sim \N(0,N')$. Here, $\{\tilde V_i\}$ and $\{W_i\}$ are independent i.i.d.~Gaussian sequences with variances 
% FIXME: space can be saved here
%$Q' = Q^2/(Q+\sigma_u^2)$ and $N' = Q \sigma_u^2/(Q+\sigma_u^2)$.
\begin{align*}
	Q' &= \frac{Q^2}{Q+\sigma_u^2},\\
	N' &= \frac{Q \sigma_u^2}{Q+\sigma_u^2}.
\end{align*}
This equivalent system model is depicted in Figure~\ref{fig:stateAmpWithNoisyObs_alt}. Note that the channel state $S^n$ is decomposed into a perfectly observable part $\tilde V^n$ and a completely unobservable part $W^n$. However, the receiver still aims to estimate the entire state $\tilde V^n + W^n$. 

\begin{figure}[t]
	\centering
	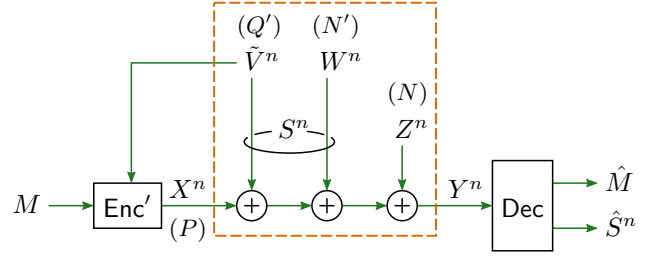
	\caption{Equivalent system model.}
	\label{fig:stateAmpWithNoisyObs_alt}
\end{figure}

\vspace{3mm}
\section{Inner bound}  \label{sec:ib}

\noindent For a constant $\beta \in [0,1]$, let
\newcommand{\negsp}{\hspace{-0.8mm}}
\begin{align}
	\hspace*{-2mm}g &= \sqrt{(1-\beta )P/Q'}, \notag  \\
	\hspace*{-2mm}r &= \begin{bmatrix}
		(1\! +\! g)Q'+N' \\
		\alpha(1\! +\! g)Q'
	\end{bmatrix}, \label{eq:ach_rdef} \\
	\hspace*{-2mm}\Sigma &= \begin{bmatrix}
		(1\negsp+\negsp g)^2Q'+\beta  P + N'+N & \beta P + \alpha(1\negsp +\negsp g)^2Q'\\
		\beta P + \alpha(1\negsp +\negsp g)^2Q' &  \beta P + \alpha^2(1\negsp +\negsp g)^2Q'
	\end{bmatrix}.  \label{eq:ach_Rdef}
\end{align}
Then we have the following inner bound to the rate--state-distortion region.
\begin{thm} \label{thm:ib}
The rate--state-distortion $(R,D)$ is achievable if
\begin{align}
	R &< \tfrac 1 2 \log \tfrac 
	{\beta  P ( \beta  P + (1+g)^2 Q' +N' + N)}
	{(N'+N)(\beta  P +\alpha^2(1+g)^2 Q') + (1-\alpha)^2 \beta  (1+g)^2 P Q'}
	, \label{eq:ach_rate} \\
	D &\geq Q'+N' - r\T \Sigma^{-1} r, \label{eq:ach_distortion}
\end{align}
for some  $\alpha \in \Real_+$, $\beta  \in [0,1]$.
\end{thm}

\begin{proof}
Consider the following coding scheme, in which the transmit signal consists of a scaled version of the channel state sequence and a Gelfand--Pinsker codeword matched to the amplified state sequence. More formally, we generate a codebook as follows.

\subsubsection*{Codebook generation} Fix $\alpha$ and $\beta$. Let $\tilde X \sim \N(0,\beta  P)$ be independent of the state $\tilde V$, and $U = \tilde X + \alpha (1+g) \tilde V$. Define an auxiliary rate $\tilde R \geq R$. For each message $m \in \natSet{2^{nR}}$, generate a subcodebook $\Cc(m)$ consisting of $2^{n(\tilde R - R)}$ sequences $u^n(m,l)$, for $l\in \natSet{2^{n(\tilde R-R)}}$, each independently generated according to $\prod_{i=1}^n p_U(u_i)$.
\subsubsection*{Encoding}
Fix $\eps'>0$. To communicate message $m$ given the state sequence $\tilde v^n$, the sender finds an index $l$ such that $(u^n(m,l), \tilde v^n) \in \Typprime(U,\tilde V)$, and transmits
\begin{align*}
	x_i &= g \tilde v_i + \tilde x_i, \quad \text{for $i \in \natSet n$},
\end{align*}
where $\tilde x_i = u_i(m,l) - \alpha(1+g) \tilde v_i$.
Observe that by construction, the sequence $\tilde x^n$ is jointly typical with the state sequence $\tilde v^n$, that is, symbolwise pairs $(\tilde x_i, \tilde v_i)$ are asymptotically independent. Therefore, the average power constraint $P$ is satisfied. 

\subsubsection*{Decoding and analysis of the probability of error}
Observe $y^n$. Let $\eps > \eps'$. Declare that message $\hat m$ has been sent if $(\hat m,\hat l)$ is the unique index pair such that
\begin{align*}
	(u^n(\hat m,\hat l), y^n) \in \Typ.
\end{align*}
By the result of Gelfand--Pinsker for channels with state~\cite{Gelfand--Pinsker1980a}, the probability of decoding error vanishes as $n\to \infty$ if 
\begin{align*}
	R &< I(U;Y) - I(U; \tilde V).
\end{align*}
To evaluate the terms, recall that 
\begin{align}
	Y &= (1+g) \tilde V + \tilde X + W + Z, \label{eq:Y} \\
	U &= \tilde X + \alpha(1+g) \tilde V, \label{eq:U} 
\end{align}
where $\tilde V$, $\tilde X$, $W$, and $Z$ are independent Gaussians of variances $Q'$, $\beta  P$, $N'$, and $N$. It is not hard to see that the rate condition evaluates to~\eqref{eq:ach_rate}.

\subsubsection*{Estimation and analysis of estimation error}
For each symbol time $i \in \natSet n$, construct the best MSE estimate of $s_i$ given the observations $y_i$ and $u_i(\hat m,\hat l)$, where $\hat m$ and $\hat l$ are the decoded message and subcode index, respectively. Omitting the time index, recall~\eqref{eq:Y}, \eqref{eq:U} and 
\begin{align}
	S &= \tilde V + W,  \label{eq:S}
\end{align}
Hence the joint distribution of $(S, Y, U)$ is 
\begin{align*}
	\begin{bmatrix}
		S\\ Y \\U
	\end{bmatrix}
	\sim \N \Biggl( 0, \begin{bmatrix}Q'+N' & r\T \\ r & \Sigma \end{bmatrix}
	\Biggr),
\end{align*}
where $r$ and $\Sigma$ are shorthand for the cross-correlation vector between $S$ and $[Y,U]$, and the autocorrelation matrix of $[Y,U]$, respectively, and are given in~\eqref{eq:ach_rdef} and~\eqref{eq:ach_Rdef}.
The best estimate is 
\begin{align}
	\hat s_i &= r\T \Sigma^{-1} \begin{bmatrix} y_i \\ u_i(\hat m,\hat l) \end{bmatrix},
	\label{eq:shat}
\end{align}
and the mean square error satisfies~\eqref{eq:ach_distortion}.
This concludes the proof of Theorem~\ref{thm:ib}.
\end{proof}

\begin{remark}
	Recall that in Costa's dirty paper coding~\cite{Costa1983}, the parameter $\alpha$ is chosen as 
	\begin{align} 
		\alpha &= \frac{\beta  P}{\beta  P + N}.
		\label{eq:costasChoice}
	\end{align}
	It turns out that varying $\alpha$ as in the theorem achieves a larger inner bound in our setting.
\end{remark}

\begin{remark}
The estimate $\hat s_i$ in~\eqref{eq:shat} is independent of $u_i$ when %$I(U;S|Y)=0$. This occurs when 
the second component of $r\T \Sigma^{-1}$ is zero,
\begin{align*}
	\begin{bmatrix} (1+g)Q'+N' \\ \alpha(1+g)Q' \end{bmatrix}\T
	\begin{bmatrix} -\beta P - \alpha(1+g)^2Q') \\ (1+g)^2Q'+\beta  P + N'+N \end{bmatrix} &= 0,
\end{align*}
which occurs when 
\begin{align}
	\alpha &= \frac{\beta  P ((1+g)Q' + N')}{(1+g)Q'(\beta P + N - gN')}.
	\label{eq:uIsUseless}
\end{align}
At all other values of $\alpha$, $u_i$ is useful in estimating $s_i$.
%\footnote{
As a side note, observe that if there is no state observation noise (when $N'=0$),~\eqref{eq:uIsUseless} reduces to~\eqref{eq:costasChoice}.
%}
\end{remark}

\begin{remark}  Numerical computations indicate that the inner bound in Theorem~\ref{thm:ib} cannot be improved by (1) reducing the gain $g$, (2) using part of the message rate to send a digital description of $\tilde V^n$, or (3) diverting a fraction of power to send an additional message by superposition coding.
\end{remark}

\vspace{3mm}
\section{Outer bounds}  \label{sec:ob}
We present two outer bounds to the rate--state-distortion region. Recall $\lambda = Q/(Q+\sigma_u^2)$. 
\begin{thm} \label{thm:ob1}
If a rate--state-distortion pair $(R,D)$ is achievable, then for all $\Nb \in [0,N]$, there exists an $\bar r \in [0, \sqrt{P (Q+\sigma_u^2)}]$ such that
\begin{align}
	\hspace*{-2mm} R &\leq \frac 1 2 \log \frac
		{(Q+\sigma_u^2)(P+N + \Es) - {\bar r}^2}
		{(Q+\sigma_u^2) (N-\Nb)},
	\label{eq:ob1_Rbound} \\
	\hspace*{-2mm} D &\geq \left(1 + 2^{2R} \cdot \frac{Q(\Nb+\sigma_u^2)(N-\Nb)}{\Nb \sigma_u^2(P + Q + N + 2 \lambda \bar r)} \right) \Es,
	\label{eq:ob1_Dbound}
\end{align}
where
\begin{align*}
	\Es &= \frac {Q \Nb \sigma_u^2}{Q\Nb + Q \sigma_u^2 + \Nb \sigma_u^2}.
\end{align*}
\end{thm}

\begin{proof}
Following~\cite{Tian2012}, let us divide the noise $Z_i$ into two independent components $\Zb_i$ and $\Zbb_i$, with 
\begin{align*}
	\Zb_i &\sim \N(0,\Nb), \\
	\Zbb_i &\sim \N(0,N-\Nb).
\end{align*}
Note that $\Es$, as defined in the theorem, denotes the mean square error of the best linear estimator of $S_i$ given $V_i= S_i+U_i$ and $S_i + \Zb_i$.
Define the quantity 
\begin{align*}
	\Delta &= \tfrac 1 n I(V^n, S^n+\Zb^n; Y^n).
\end{align*}
Using results from the remote source coding problem in rate--distortion theory, it was shown in~\cite{Tian2012} that %for any choice $\Nb$, 
the information measure $\Delta$ and the quadratic distortion $D$ are related as
\begin{align*}
	\Delta &\geq \tfrac 1 2 \log\left( \frac 
		{Q(\Nb+\sigma_u^2)}
		{\Nb \sigma_u^2 (D/\Es - 1)}
	\right),
\end{align*}
or equivalently,
\begin{align}
	D &\geq \left( 1 + 2^{-2\Delta} \cdot \frac{Q(\Nb+\sigma_u^2)}{\Nb \sigma_u^2} \right) \Es
	. \label{eq:ob1_DeltaToD}
\end{align}
This allows us to translate upper bounds on $\Delta$ to lower bounds on $D$. %similar to the case with perfect state observation in~\cite{Kim2008}. 
Next, we obtain an outer bound on the achievable $(R,\Delta)$ region. First, it follows from Fano's inequality that
\begin{align}
	nR &= H(M) \notag \\
	&= H(M \cond V^n, S^n+\Zb^n) \notag \\
	&\leq I( M; Y^n \cond V^n, S^n+\Zb^n) + n\eps_n \notag \\	
	&= h(Y^n \cond V^n, S^n+\Zb^n) - h( \Zbb^n ) + n\eps_n.  \label{eq:ob1_RFano}
\end{align}
Let
\begin{align}
	\bar r &= \frac 1 n \sum_{i=1}^n \abs{\E(X_i V_i)}.  \label{eq:barr}
\end{align}
Using covariance matrices to bound differential entropies, 
%we show in Appendix~\ref{sec:BoundsByCovariance} that 
it can be shown that
\begin{align}
	h(Y^n)  
	&\leq \tfrac n 2 \log \bigl( 2\pi e (P \! + \! Q \! + \! N \! + \! 2\lambda \bar r ) \bigr), \label{eq:ob2_hY}
	\\
	\hspace*{-2mm} h(Y^n \cond V^n, S^n \! + \! \Zb^n) 
	&\leq \tfrac n 2 \log\bigl( 2 \pi e \bigl(
		P+N + \Es \notag \\
		&\hspace{21mm} - {\bar r}^2/(Q+\sigma_u^2)
	\bigr) \bigr).
	\label{eq:ob2_hYcond}
\end{align}
Substituting these inequalities and the definition of $\Delta$ into~\eqref{eq:ob1_RFano}, we conclude that $(R,\Delta)$ must satisfy~\eqref{eq:ob1_Rbound} and
\begin{align*}
	\Delta &\leq \tfrac 1 2 \log \frac{ P + Q + N + 2 \lambda \bar r}{ N-\Nb}  - R
	%\label{eq:ob1_RDeltabound}
\end{align*}
as $n \to \infty$. 
Using the last inequality with~\eqref{eq:ob1_DeltaToD}, we obtain~\eqref{eq:ob1_Dbound}.

Finally, note that $\bar r \in [0, \sqrt{P (Q+\sigma_u^2)}]$ since
\begin{align*}
	\bar r 
	&\annleq{a} \sqrt{\frac 1 n \sum_{i=1}^n \E(X_i V_i)^2} \\
	&\annleq{b} \sqrt{ \frac 1 n \sum_{i=1}^n \E(X_i^2) (Q+\sigma_u^2) } \\
	&\annleq{c} \sqrt{P (Q+\sigma_u^2)},
\end{align*}
where (a) follows from the generalized mean inequality, (b) follows from the Cauchy--Schwarz inequality, and (c) follows from the power constraint. This concludes the proof of the outer bound in Theorem~\ref{thm:ob1}.
\end{proof}

\vspace{4mm}
In order to state the second outer bound, we define the function $f$ as 
\begin{align*}  %FIXME: space: can convert to inline
	f(x) &= (\sqrt{x} - \sqrt{N'/Q'} \sqrt{Q'-x})_+^2,
\end{align*}
where $x_+$ denotes the positive part $\max\{x,0\}$. It can be shown that
%In Appendix~\ref{sec:fConvexMonotone}, we show that 
$f$ is convex and non-decreasing.
%
%Then we have the following outer bound.
\begin{thm} \label{thm:ob2} 
	If a rate--distortion pair $(R,D)$ is achievable, then it must satisfy 
	\begin{align}
		R &\leq \frac 1 2 \log \frac
			{\sigma_u^2 (N+P+Q)  + Q(N+P)-{\bar r}^2}
			{(Q+ \sigma_u^2)(N+N')}
		, \label{eq:ob2_Rbound} \\
		D &\geq f\left( \frac{Q' (N+N')}{P+Q+N+2\lambda\bar r} \ 2^{2R} \right)
			\label{eq:ob2_Dbound}
	\end{align}
	for some $\bar r \in [0, \sqrt{P (Q+\sigma_u^2)}]$.
\end{thm}

\begin{proof}%[Proof of Theorem~\ref{thm:ob2}]
Let $\Delta' = \frac 1 n I(\tilde V^n; Y^n)$. To relate $\Delta'$ and the distortion $D$,  note that
\begin{align*}
	\Delta' 
	&\anngeq{a} \tfrac 1 n I(\tilde V^n; \hat S^n) \notag \\
	&= \tfrac 1 n h(\tilde V^n) - \tfrac 1 n h(\tilde V^n \cond \hat S^n) \notag \\
	&\geq \tfrac 1 n h(\tilde V^n) - \tfrac 1 n \sum_{i=1}^n h(\tilde V_i \cond \hat S_i) \notag \\
	&\geq \tfrac 1 n h(\tilde V^n) - \tfrac 1 {2n} \sum_{i=1}^n \log( 2\pi e \Var(\tilde V_i \cond \hat S_i) ) \notag \\
	&\anngeq{b} \tfrac 1 n h(\tilde V^n) - \tfrac 1 2 \log \Bigl( 2\pi e \cdot \tfrac 1 n \sum_{i=1}^n \Var(\tilde V_i \cond \hat S_i) \Bigr) \notag \\
	&\anngeq{c} \tfrac 1 n h(\tilde V^n) - \tfrac 1 2 \log \Bigl( 2\pi e \cdot \tfrac 1 n \sum_{i=1}^n D'_i \Bigr) \notag \\
	&= \tfrac 1 2 \log(Q')- \tfrac 1 2 \log \Bigl(\tfrac 1 n \sum_{i=1}^n D'_i \Bigr),
\end{align*}
where (a) follows from the data processing inequality, (b) follows from Jensen's inequality and (c) follows from defining $D'_i$ to be the mean square error of best \emph{linear} estimator of $\tilde V_i$ given $\hat S_i$. Hence
\begin{align}
	\tfrac 1 n \sum_{i=1}^n D'_i &\geq Q' \cdot 2^{-2 \Delta'}.
		\label{eq:ob2_DeltaprimeToDprime}
\end{align}
Let $D_i = \Var( S_i - \hat S_i)$ be the mean square instantaneous state estimation error at the receiver at time $i$.
In Appendix~\ref{sec:DprimeToD}, we use the Cauchy--Schwarz inequality to show
\begin{align*}
	D_i &\geq f(D'_i).
\end{align*}
Thus,
\begin{align}
	D = \tfrac 1 n \sum_{i=1}^n D_i 
	& \geq \tfrac 1 n \sum_{i=1}^n f(D'_i) \notag \\
	& \anngeq{a} f \left( \tfrac 1 n \sum_{i=1}^n D'_i \right) \notag \\
	& \anngeq{b} f( Q' \cdot 2^{-2 \Delta'} ), \label{eq:ob2_DeltaprimeToD}
\end{align}
where (a) follows from Jensen's inequality and (b) follows from~\eqref{eq:ob2_DeltaprimeToDprime}.

Next, we obtain an outer bound on the achievable $(R,\Delta')$ region. From Fano's inequality, we have
\begin{align}
	nR 
	&= H(M \cond \tilde V^n) \notag \\
	&\leq I( M; Y^n \cond \tilde V^n) + n\eps_n \notag \\	
	&= h(Y^n \cond \tilde V^n) - \underbrace{h( Y^n \cond \tilde V^n, M )}_{=\tfrac n 2 \log(2\pi e(N'+N))} + n\eps_n.  \label{eq:ob2_RFano}
\end{align}
Reusing definition~\eqref{eq:barr}, and using covariance matrices to bound differential entropies, it can be shown that
\begin{align*}
	h(Y^n \cond \tilde V^n) 
	&\leq \tfrac n 2 \log\Bigl( 2 \pi e \frac
		{ \sigma_u^2 (N\!+\!P\!+\!Q)  + Q(N\!+\!P) -{\bar r}^2 }
		{Q+ \sigma_u^2}
	\Bigr)
	.
\end{align*}
%see Appendix~\ref{sec:BoundsByCovariancePrime} for details.
Substituting this inequality, bound~\eqref{eq:ob2_hY} and the definition of $\Delta'$ into~\eqref{eq:ob2_RFano}, we obtain~\eqref{eq:ob2_Rbound} and
\begin{align*}
	\Delta' &\leq \tfrac 1 2 \log \frac{ P + Q + N + 2 \lambda \bar r}{ N+N'} - R %\label{eq:RDeltaprime_bound}
\end{align*}
as $n\to\infty$.
Using the last inequality with~\eqref{eq:ob2_DeltaprimeToD} yields~\eqref{eq:ob2_Dbound}, which concludes the proof of Theorem~\ref{thm:ob2}.
\end{proof}

\vspace{3mm}
\section{Numerical examples and Concluding remarks} \label{sec:numExConclusion}

Two examples for the inner bound of Theorem~\ref{thm:ib} and the outer bounds of Theorems~\ref{thm:ob1} and~\ref{thm:ob2} are depicted in Figures~\ref{fig:numExample_lowPower} and~\ref{fig:numExample_highPower}. In the example in Figure~\ref{fig:numExample_highPower}, the transmit power $P$ is sufficiently large to permit a nonzero rate at the minimum distortion value. This is the case where the optimal transmission scheme without a message ($R=0$) as discussed in~\cite{Tian2012} uses less than the full transmit power.

We observe that the outer bounds complement each other, but exhibit a nonnegligible gap from the inner bound. To close this gap, new outer bounding techniques will be necessary, the study of which is the subject of future research.

\begin{figure}
	\centering
	\vspace*{2mm}
	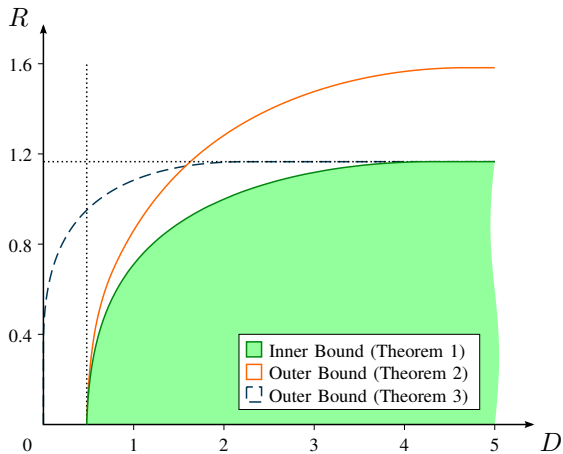
	\caption{Bounds to the rate--estimation-error region, parameters $Q=10$, $N=1$, $\sigma_u^2=1$, and $P=7.7$ (i.e., $Q' = 9.09$, $N' = 0.91$).}
	\label{fig:numExample_lowPower}
\end{figure}

\begin{figure}
	\centering
	\vspace*{2.5mm}
	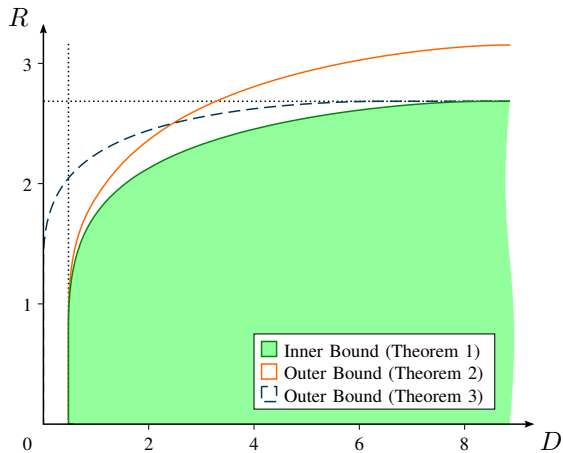
	\caption{Bounds to the rate--estimation-error regions, parameters $Q=10$, $N=1$, $\sigma_u^2=1$, and $P=77$ (i.e., $Q' = 9.09$, $N' = 0.91$).}
	\label{fig:numExample_highPower}
\end{figure}

% \vspace{3mm}
% \section*{Acknowledgment} 

%\clearpage
\appendix

\subsection{Relation between $D_i$ and $D'_i$} \label{sec:DprimeToD}
Let $E_i = S_i - \hat S_i$ and recall $D_i=\Var( S_i - \hat S_i)=\Var(E_i)$. Likewise, let $\hat V_i$ be the best linear estimator of $\tilde V_i$ given $\hat S_i$, define the corresponding estimation error as $F_i = \tilde V_i - \hat V_i$ and recall $D'_i = \Var(F_i)$.

By the orthogonality principle, we have $\E( F_i \hat V_i ) = 0$, and using $\E( \tilde V_i^2 ) = Q'$, it follows that
\begin{align*}
	\E(\tilde V_i \hat V_i ) & = \E(\hat V_i^2) = Q' - D'_i.
\end{align*}
Since $\hat V_i$ is a scaled version of $\hat S_i$, the variance $D_i$ is lower bounded by the mean square error of the best linear estimator of $S_i$ given $\hat V_i$, namely
\begin{align}
	D_i &\geq Q'+N' - \frac{\E(\hat V_i S_i)^2}{Q'-D'_i}.
		\label{eq:DprimeToD_DiLB}
\end{align}
To evaluate the expectation, consider 
\begin{align*}
	\hat V_i &= \frac{Q'-D'_i}{Q'} \, \tilde V_i + F'_i,
\end{align*}
where the first term in the sum is the best linear estimator of $\hat V_i$ given $\tilde V_i$, and $F'_i$ is the corresponding estimation error. Thus, $\tilde V_i$ and $F'_i$ are uncorrelated and
\begin{align*}
	\Var( F'_i ) &= Q' - D'_i - \frac{(Q'-D'_i)^2}{Q'} = \frac{D'_i(Q'-D'_i)}{Q'}.
\end{align*}
Furthermore, recall
\begin{align*}
	S_i &= \tilde V_i + W_i,
\end{align*}
where $\tilde V_i$ and $W_i$ are independent and $\Var(W_i) = N'$. Thus,
\begin{align}
	\E(\hat V_i S_i) 
	&= \E\left( \left(\frac{Q'-D'_i}{Q'} \, \tilde V_i + F'_i \right)(\tilde V_i + W_i) \right) \notag \\
	&= Q'-D'_i + \E(F'_i W_i).
		\label{eq:DprimeToD_EVtS}
\end{align}
By the Cauchy--Schwarz inequality, 
\begin{align*}
	\abs{ \E(F'_i W_i) } &\leq \sqrt{ \E({F'_i}^2) \E(W_i^2) } \\
	&= \sqrt{ \frac{N' D'_i (Q'-D'_i)}{Q'} }
\end{align*}
Substituting back into~\eqref{eq:DprimeToD_EVtS} and using $D'_i\leq Q'$, this implies
\begin{align*}
	\abs{\E(\hat V_i S_i)}
	&\leq Q'-D'_i + \sqrt{ \frac{N' D'_i (Q'-D'_i)}{Q'} }.
\end{align*}
Further, substituting back into~\eqref{eq:DprimeToD_DiLB} yields
\begin{align*}
	D_i &\geq Q'+N' - \frac{\left(Q'-D'_i + \sqrt{ \frac{N' D'_i (Q'-D'_i)}{Q'} }\right)^2}{Q'-D'_i} \notag \\
	&= \left( \sqrt{D'_i} - \sqrt{N'/Q'} \sqrt{Q'-D'_i} \right)^2   \notag \\
	&\geq \left( \sqrt{D'_i} - \sqrt{N'/Q'} \sqrt{Q'-D'_i} \right)_+^2 \notag \\
	&= f(D'_i),
\end{align*}
which concludes the proof.

\vfill
\bibliographystyle{IEEEtran}
\bibliography{IEEEabrv,references-unified} %,../unified-bibtex/myPublications}

\end{document}